\documentclass[12pt]{article}
\usepackage{amsfonts, amsthm, amsmath, amssymb}
\usepackage{caption,array,verbatim,enumitem}
\usepackage[T1]{fontenc}
\usepackage{tikz}
\usepackage{times}
\usepackage{geometry}
\usepackage[mathscr]{euscript}
\usepackage{bbm}
\usepackage{color,soul}

\usepackage{nicefrac}
\usepackage[nameinlink,capitalise]{cleveref}

\usepackage{scalerel}
\SetLabelAlign{fixedwidth}{\hss\llap{\makebox[-0.5em][l]{#1}}}
\usepackage{relsize}

\usepackage[english]{babel}
\usepackage[utf8]{inputenc}
\usepackage{mathtools}

\newif\ifnotes
\notestrue
\ifnotes

\makeatletter
\newcommand{\stackalign}[1]{
  \vcenter{
    \Let@ \restore@math@cr \default@tag
    \baselineskip\fontdimen10 \scriptfont\tw@
    \advance\baselineskip\fontdimen12 \scriptfont\tw@
    \lineskip\thr@@\fontdimen8 \scriptfont\thr@@
    \lineskiplimit\lineskip
    \ialign{\hfil$\m@th\scriptstyle##$&$\m@th\scriptstyle{}##$\crcr
      #1\crcr
    }
  }
}

\makeatletter
\def\legendre@dash#1#2{\hb@xt@#1{%
  \kern-#2\p@
  \cleaders\hbox{\kern.5\p@
    \vrule\@height.2\p@\@depth.2\p@\@width\p@
    \kern.5\p@}\hfil
  \kern-#2\p@
  }}
\def\@legendre#1#2#3#4#5{\mathopen{}\left(
  \sbox\z@{$\genfrac{}{}{0pt}{#1}{#3#4}{#3#5}$}%
  \dimen@=\wd\z@
  \kern-\p@\vcenter{\box0}\kern-\dimen@\vcenter{\legendre@dash\dimen@{#2}}\kern-\p@
  \right)\mathclose{}}
\newcommand\legendre[2]{\mathchoice
  {\@legendre{0}{1}{}{#1}{#2}}
  {\@legendre{1}{.5}{\vphantom{1}}{#1}{#2}}
  {\@legendre{2}{0}{\vphantom{1}}{#1}{#2}}
  {\@legendre{3}{0}{\vphantom{1}}{#1}{#2}}
}
\def\dlegendre{\@legendre{0}{1}{}}
\def\tlegendre{\@legendre{1}{0.5}{\vphantom{1}}}
\makeatother

\def\namedlabel#1#2{\begingroup
	#2%
	\def\@currentlabel{#2}%
	\phantomsection\label{#1}\endgroup
}
\newcommand{\E}{\mathbb{E}}
\newcommand{\F}{\mathbb{F}}

\newcommand{\N}{\mathbb{N}}

\newcommand{\calC}{\mathcal{C}}

\newcommand{\calO}{\mathcal{O}}

\newcommand{\rw}{{\sf RW}}
\newcommand{\sw}{s{\sf RW}}

\newcommand{\boldone}{{1\!\!1}}

\newcommand{\ep}{\varepsilon}

\newcommand{\prob}[1]{{\rm Pr}_{#1}}
\newcommand{\R}{\mathbb{R}}

\newcommand{\backwardsvec}[1]{\reflectbox{\ensuremath{\vec{\reflectbox{\ensuremath{#1}}}}}}

\newtheorem{lemma}{Lemma}

\newtheorem{fact}{Fact}
\newtheorem{clm}{Claim}

\newtheorem{defn}{Definition}
\newtheorem{thm}{Theorem}

\title{Analyzing Ta-Shma's Code via the Expander Mixing Lemma %\thanks{Supported by UC Lab Fees grant LFR-18-548554. All opinions expressed are those of the authors.}
}
\author{
Silas Richelson\thanks{UC Riverside. Email:
  \texttt{silas@cs.ucr.edu}.} \and
  Sourya Roy\thanks{UC Riverside. Email:
 \texttt{sourya.roy@email.ucr.edu}.}
}
\allowdisplaybreaks

\begin{document}
\date{}
\maketitle
\setcounter{page}{0}

\thispagestyle{empty}
\begin{abstract}
Random walks in expander graphs and their various derandomizations (\emph{e.g.}, replacement$/$zig-zag product) are invaluable tools from pseudorandomness.  Recently, Ta-Shma used $s$-wide replacement walks in his breakthrough construction of a binary linear code almost matching the Gilbert-Varshamov bound (STOC 2017).  Ta-Shma's original analysis was entirely linear algebraic, and subsequent developments have inherited this viewpoint.  In this work, we rederive Ta-Shma's analysis from a combinatorial point of view using repeated application of the \emph{expander mixing lemma}.  We hope that this alternate perspective will yield a better understanding of Ta-Shma's construction. As an additional application of our techniques, we give an alternate proof of the \emph{expander hitting set lemma}.
\end{abstract}

\newcommand{\calc}{{\mathcal{C}}}

\section{Introduction}
Error correcting codes (ECCs) allow a sender to encode a message so that the receiver can recover the full message even if several codeword bits are lost or flipped during transmission.  ECCs are incredibly useful, both in theory and in practice~\cite{SHAMIR79,STV01,ChenW19} (and many, many more).  Formally, a binary code is a map $\calC: \{0,1\}^{k} \to \{0,1\}^n$ which sends a message $m \in \{0,1\}^k$ to the codeword $\calC(m) \in \{0,1\}^n$. Two important parameters of a code are the \emph{distance} and \emph{rate}, which are respectively measures of the code's quality and efficiency. \emph{Rate} is the ratio $k/n$, the number of message bits per codeword bit while \emph{distance} refers to the minimum fraction of coordinates (in $[n]$)  on which two distinct codewords disagree. One of the holy grails in coding theory is to find the best tradeoff between the distance and rate of a binary code. It is known that codes with optimal distance $\delta =\nicefrac12$ must have exponentially small rate~\cite{PLOT60}.  The Gilbert-Varshamov (GV) bound~\cite{Gil52,Var57} states for any $\delta \in (0,\nicefrac12)$, there exists a code $C_n$ with blocklength $n$ and distance $d$ with rate $1-H(\delta) - o_n(1)$ where $H(\cdot)$ is Shannon's binary entropy function. Unfortunately, this is a probabilistic (or greedy) construction and we do not know of explicit binary codes matching this bound. For distances $\delta$ close to $\nicefrac12$, the GV bound states that there exists a code with distance $\nicefrac{(1-\ep)}{2}$ and rate $\Omega(\ep^2)$. On the other hand, it is known that any code with distance $\nicefrac{(1-\ep)}{2}$ must have rate $\mathcal{O}\bigl(\ep^2\cdot\log(1/\ep)\bigr)$~{\cite{ABNNR92}}. Constructing an explict code matching the GV bound even for these distance parameters is a major open problem. 

A few years ago, in a breakthrough result, Ta-Shma~\cite{TaShma17} described an explicit construction which got very close: he constructed a family of codes $\{C_n\}_n$ with rate $\Omega(\ep^{2+o_{\ep}(1)})$ and distance $\nicefrac{(1-\ep)}{2}$.  The core of his construction is an amplification procedure which increases the distance of the code using certain special types of random walks on expander graphs.  Specifically, Ta-Shma encodes a message $m\in\{0,1\}^k$ as follows.
\begin{enumerate}
    \item Use a ``base code''  $\calC_0:\{0,1\}^k \to \{0,1\}^n$ with a good (but not optimal) rate$/$distance tradeoff, to encode message $m\in\{0,1\}^k$ into a $n$-bit codeword $\calC_0(m)$ which we will equivalently interpret as function $f:[n]\rightarrow\{0,1\}$.%\footnote{For example, using Justesen codes~\cite{Just} ensures $n\leq25k$ (so rate is at least $.04$) and distance is at least $.045$.}
    
    \item Identify the coordinate set $[n]$ with the vertices of an expander graph $A$.\footnote{We abuse notation by refering to $A$ both as the graph and the vertex set.}
    
    \item Let $W\subset A^t=[n]^t$ be a \emph{special} subset of the set of all $t$-length walks in $A$.  Define $g:W\rightarrow\{0,1\}$ by $g(a_1,\dots,a_t)=f(a_1)\oplus\cdots\oplus f(a_t)$, where $\oplus$ is the bit XOR.  Output $g\in\{0,1\}^{|W|}$.
\end{enumerate}
The ingenious component in TaShma's construction is the choice of the subset $W$. 
As we will soon see, choosing $W$ to be the set of all $t$-length walks in $A$ does not yield an optimal distance$/$rate tradeoff.  TaShma, instead, uses a derandomized subset of walks, resulting from taking an \emph{$s$-wide replacement product walk} on $A$.  In the ordinary replacement product, another expander $B$ is chosen with $|B|={\rm deg}(A)$ so that given $a\in A$, each $b\in B$ corresponds to some $a'\in N(a)$.  A $t$-length replacement product walk in $A$ chooses a random $a\sim A$ and a $(t-1)$-length walk $(b_1,\dots,b_{t-1})$ in $B$ and outputs the walk $(a_1,\dots,a_t)$ in $A$ where $a_1=a$ and $a_{i+1}$ is the $b_i$-th neighbor of $a_i$ for $i=1,\dots,t-1$.  Note the set of replacement product walks in $A$ is a proper subset of the set of all walks.  The $s$-wide replacement product is a parametrized version of the ordinary replacement product.  We explain the $s$-wide replacement product in detail in Section~\ref{sec:prelims}.

\subsection{Our Contribution}
In this note, we rederive the analysis of TaShma~\cite{TaShma17} using repeated applications of the \emph{Expander Mixing Lemma}. TaShma's original analysis, as well as subsequent developments, convey a strongly linear algebraic viewpoint. In this writeup, we take the expander mixing lemma as our starting point and proceed from there in a combinatorial fashion.  Thus, we demonstrate that no linear algebra is needed for the analysis of Ta-Shma's code beyond that which is needed to prove the expander mixing lemma.  We would like to be forthcoming and stress that \textbf{our analysis is completely equivalent to Ta-Shma's original analysis}.  So if you are hoping to read about a new code with improved parameters, you should read something else.  This paper is for those researchers who have had difficulty penetrating the intuition behind Ta-Shma's construction.  We believe that this alternate perspective will appeal to a wider audience and make it easier for the scientific community to innovate on Ta-Shma's breakthrough work.

Our proof is the same as the original proof insofar as a random walk on a graph can be modelled both as a random process and as a linear operator.  The original analysis takes the linear operator view, we take the random process view.  In theory, the linear operator view is convenient for quantitatively reason about random walks because it reduces the task to understanding repeated multiplication by a fixed matrix.  However, when analyzing replacement product walks from the linear operator perspective, the adjacency matrices of the outer and inner expander graphs have to be combined using some kind of tensor product.  The situation is worse for the $s-$wide replacement product since then one has to keep track of $s$ different tensor product matrices and the iterated matrix product needs to alternate over these $s$ matrices.  Thus, it seems there are diminishing returns in terms of the simplicity afforded by the linear operator perspective when the set of all random walks is to be derandomized.  By using the random process view, we are able to express the same ideas in a much simpler way.  This, in turn, makes it easier to see what is going on in certain key steps of the argument.

% the expander mixing lemma (and thus, we are implicitly hiding the linear algebra needed to prove this simple 

% Our work shows that the only linear algebra needed is that needed to prove the expander-mixing-lemma. To this end, we define a expander in terms of the parameters of the expander-mixing-lemma (see \cref{def:expander}).  We hope that this alternate perspective of TaShma's analysis will yield a better understanding of TaShma's codes and possibly lead to codes with better parameters, or to progress on decoding algorithms for these codes\footnote{Partial progress towards decoding was made by \cite{JQST20,AJQST20,JST20}, but list decoding algorithms which work near the Johnson bound are not currently known.}.

\subsection{Techniques: Expander Mixing Lemma and consequences}
\label{sec:techniques}
\paragraph{Notation.} Throughout this paper, we refer to graphs by their vertex sets, and use $\sim$ to indicate that two vertices are connected with an edge.  So for example, if $A$ is a graph and $a,a'\in A$ are vertices, we write $a\sim a'$ if there is an edge between $a$ and $a'$.  We write $\rw_A^t$ (resp. $\rw_A^t(a)$) for the distribution which outputs a $t-$length random walk in $A$ (resp. a $t-$length random walk in $A$ which begins at $a$). Given two distributions $\mathcal{D}$ and  $\mathcal{D}'$, we will write $\mathcal{D}\equiv \mathcal{D}'$ to denote that they are same.  

\vskip 3mm In order to get a sense for our technique, let us analyze the distance amplification procedure resulting from taking a random walk on an expander.  Typically expander graphs are defined via the second largest eigenvalue of the adjacency matrix of the graph; in this paper we will use the following equivalent definition (similar definitions have been used in other works, \emph{e.g.},~\cite{DK17}).

\begin{defn}
\label{def:expander}
We say that a graph $A$ is a $\lambda-$\emph{expander} if for all $f,g:A\rightarrow\R$, the following holds: \[\Big|\E_{a\sim a'}\bigl[f(a)\cdot g(a')\bigr]-\mu_f\mu_g\Big|\leq\lambda\sigma_f\sigma_g,\] where $\mu_f$ and $\sigma_f$ are the expectation and standard deviation of the random variable $f(a)$ (namely, $\mu_f =\E_a\bigl[f(a)\bigr]$ and $\sigma_f^2+\mu_f^2=\E_a\bigl[f(a)^2\bigr]$, and similarly for $\mu_g$ and $\sigma_g$).
\end{defn}

Now consider the distance amplification framework above instantiated with $A$ being a constant degree, $d-$regular $\lambda-$expander, and $W$ being the set of all $t-$length random walks in $A$.  Note that $|W|=n\cdot d^{t-1}$, and so the rate of the resulting code is $\calO(d^{-t})$.  If $A$ is Ramanujan (\emph{i.e.}, an expander with the best possible relationship between $\lambda$ and $d$) then $\lambda\approx2/\sqrt{d}$ which makes the rate $\calO\bigl((\lambda/2)^{2t}\bigr)$.  Regarding the distance, note that for any $n-$bit string $f:[n]\rightarrow\{0,1\}$, if the fraction of non-zero coordinates is $\frac{1-\ep}{2}$, then $\ep=-\E_{v\sim[n]}\bigl[(-1)^{f(v)}\bigr]$.  For this reason, we show that the amplification framework above decreases \emph{bias}, where \[{\rm Bias}(f):=\Big|\E_{v\sim[n]}\bigl[(-1)^{f(v)}\bigr]\Big|.\]  The claim below shows that when $W$ is the set of all $t-$length walks in $A$, a regular Ramanujan expander graph with expansion $\lambda$, and when ${\rm Bias}(f)\leq\sqrt{\lambda}$, then ${\rm Bias}(g)\leq\frac{1}{2}\cdot(4\lambda)^{t/2}$.  It follows that if the distance of the amplified code is $\frac{1-\ep}{2}$, then the rate is $\Omega(\ep^4\cdot8^{-2t}\bigr)$.  For any constant $\alpha>0$, it is possible to choose parameters so that $\ep^\alpha\leq8^{-2t}$, in which case the rate is $\Omega(\ep^{4+\alpha})$.

\begin{clm}
\label{clm:exprw}
Let $A$ be a regular $\lambda-$expander, $f:A\rightarrow\{0,1\}$ a function of bias $\big|\E_a\bigl[(-1)^{f(a)}\bigr]\big|\leq\sqrt{\lambda}$.  For $k\geq1$, define $h_k:A\rightarrow\R$ as \[h_k(a):=\E_{(a_1,\dots,a_k)\sim\rw_A^k(a)}\Bigl[(-1)^{f(a_1)\oplus\cdots\oplus f(a_k)}\Bigr].\]  Let $\ep_k:=\big|\E_a\bigl[h_k(a)\bigr]\big|$ and $\sigma_k$ be such that $\sigma_k^2+\ep_k^2=\E_a\bigl[h_k(a)^2\bigr]$.  Then for all $k\geq1$: \[\ep_k\leq\frac{1}{2}\cdot(4\lambda)^{k/2};\text{  }\sigma_k\leq\sqrt{\E_a\bigl[h_k(a)^2\bigr]}\leq(4\lambda)^{\frac{k-1}{2}}.\]
\end{clm}

\noindent We will actually prove the following slight generalization of Claim~\ref{clm:exprw}, which will be more useful in our analysis later on.  Note Claim~\ref{clm:exprw} is recovered from Claim~\ref{clm:genrw} by letting $H$ be the constant function which always outputs $1$, and noting that $\hat\ep_1\leq\sqrt{\lambda}$ and $\hat\sigma_1\leq1$.

\begin{clm}
\label{clm:genrw}
Let $A$ be a regular $\lambda-$expander, $f:A\rightarrow\{0,1\}$ a function of bias $\big|\E_a\bigl[(-1)^{f(a)}\bigr]\big|\leq\sqrt{\lambda}$, and $H:A\rightarrow\R$ any function.  For $k\geq1$, let $\hat h_k:A\rightarrow[0,1]$ be defined by \[\hat h_k(a)=\E_{(a_1,\dots,a_k)\sim\rw^k(a)}\Bigl[(-1)^{f(a_1)\oplus\cdots\oplus f(a_k)}\cdot H(a_k)\Bigr].\]  Let $\hat\ep_k:=\big|\E_a\bigl[\hat h_k(a)\bigr]\big|$ and $\hat\sigma_k$ such that $\hat\sigma_k^2+\hat\ep_k^2=\E_a\bigl[\hat h_k(a)^2\bigr]$.  Then for $k\geq2$, \[ \hat\ep_k\leq2^{k-2}\cdot(\lambda^{\frac{k-1}{2}}\hat\ep_1+\lambda^{\frac{k}{2}}\hat\sigma_1);\text{ and }\hat\sigma_k\leq\sqrt{\E_a\bigl[\hat h_k(a)^2\bigr]}\leq2^{k-2}\cdot(\lambda^{\frac{k-2}{2}}\hat\ep_1+\lambda^{\frac{k-1}{2}}\hat\sigma_1).\]
\end{clm}

\begin{proof}
The key observation is that for $k\geq2$, $\hat h_k(a)=(-1)^{f(a)}\cdot \E_{a'\sim N(a)}\bigl[\hat h_{k-1}(a')\bigr]$.  This lets us bound $\hat\ep_k$ and $\hat\sigma_k$ in terms of $\hat\ep_{k-1}$ and $\hat\sigma_{k-1}$ using the expander mixing lemma (Definition~\ref{def:expander}) as follows:
\begin{center}\begin{minipage}{.9\linewidth}\begin{itemize}
    \item[$\cdot$ $\hat\ep_k=$] $\big|\E_a\bigl[\hat h_k(a)\bigr]\big|=\big|\E_{a\sim a'}\bigl[(-1)^{f(a)}\cdot \hat h_{k-1}(a')\bigr]\big|\leq\sqrt{\lambda}\hat\ep_{k-1}+\lambda\hat\sigma_{k-1}$;
   \item[$\cdot$ $\hat\sigma_k^2\leq$] $\hat\sigma_k^2+\hat\ep_k^2=\E_a\bigl[\hat h_k(a)^2\bigr]=\E_a\Bigl[\E_{a'\sim N(a)}\bigl[\hat h_{k-1}(a')\bigr]^2\Bigr]=\E_{a'\sim_{A^2}a''}\bigl[\hat h_{k-1}(a')\cdot \hat h_{k-1}(a'')\bigr]$
    \item[$\leq$] $\hat\ep_{k-1}^2+\lambda^2\hat\sigma_{k-1}^2$,
\end{itemize}\end{minipage}\end{center}
where $a'\sim_{A^2}a''$ indicates that $(a',a'')$ is a uniform edge in $A^2$ (a $\lambda^2-$expander).  We have used that the distribution which draws $a\sim A$, $a',a''\sim N(a)$ and outputs $(a',a'')$ is identical to the uniform edge distribution on $A^2$.  The claim follows by induction.
\end{proof}

\subsection{Improving the rate via $s$-wide replacement product walks}
The rate of the above code is roughly $\ep^4$, which is too low. In order for it to have rate $\approx\ep^2$, we would have needed $\ep_t\leq\lambda^t$ rather than what we got which was $\ep_t\leq\lambda^{t/2}$ (actually we got something weaker, we are oversimplifying to clarify the discussion).  The recursive formulas which appeared in the proof were:
\begin{itemize}
    \item[$\cdot$] $\ep_k\leq{\rm Bias}(f)\cdot\ep_{k-1}+\lambda\sigma_{k-1}\leq\sqrt{\lambda}\ep_{k-1}+\lambda\sigma_{k-1}$ (we assumed ${\rm Bias}(f)\leq\sqrt{\lambda}$);    
    \item[$\cdot$] $\sigma_k\leq\ep_{k-1}+\lambda\sigma_{k-1}$ (implied by $\sigma_k^2\leq\ep_{k-1}^2+\lambda^2\sigma_{k-1}^2$).
\end{itemize}
The problem here is the bound $\sigma_k\leq\ep_{k-1}+\lambda\sigma_{k-1}$, specifically the $\ep_{k-1}$ term on the right since we are moving from a $k-$th level term to a $(k-1)-$th level term without gaining a factor of $\lambda$.  Plugging this into the first equation gives $\ep_k\leq\sqrt{\lambda}\ep_{k-1}+\lambda\ep_{k-2}+\lambda^2\sigma_{k-2}$, where the first two terms are problematic (we are moving from level $k$ to level $k-1$ and $k-2$ but gaining only one factor of $\sqrt{\lambda}$ and $\lambda$, respectively).  The first problematic term could be fixed by choosing $\lambda$ such that ${\rm Bias}(f)\leq\lambda$; but the second problematic term cannot be easily fixed.  This phenomenon was observed in~\cite{TaShma17} where the problem is summarized by saying ``one out of every two steps works''.

A natural idea for derandomizing $W$ is to work with a set of replacement (or zig-zag) product walks.  Unfortunately this yields no improvement as the ``one out of every two steps works'' problem persists.  Ben-Aroya and Ta-Shma~\cite{BT11} solved this problem in a different context by using an expander graph $B$ on a slightly larger vertex set of size $d^s$ for $s\geq2$, and by analyzing the resulting walk $s$ steps at a time.  This is called the $s$-wide replacement product.  Ta-Shma was then able to successfully argue that ``$s-4$ out of every $s$ steps work''.  When interpreted in our language, this observation translates to a recursive formula like $\ep_k\leq\lambda^{s-4}\cdot\ep_{k-s}$, where we move from a $k-$th level term to a $(k-s)-$th level term, while gaining $(s-4)$ factors of $\lambda$.  Gaining $s$ factors of $\lambda$ would have let us solve to the optimal $\ep_k\leq\lambda^k$, obtaining rate of $\approx\ep^2$; gaining $(s-4)$ factors of $\lambda$ lets us solve instead to $\ep_k\leq\lambda^{k(1-4/s)}$ which is almost as good when $s$ is large.

\section{Preliminaries}
\label{sec:prelims}
% \paragraph{Graphs as Markov Operators.} For analyzing Ta-Shma's construction, it will be convenient to use a probabilistic notion of a graph.  Formally, we say that a finite vertex set $A$ is a \emph{graph} if for all $a\in A$ there is a distribution $N(a)$ on $A$ called the \emph{neighborhood distribution of} $a$.  The \emph{adjacency matrix} of a graph $A$ with $|A|=n$ is the matrix $M\in\R^{n\times n}$ where $M_{a,a'}:=p_a(a')$ where $p_a$ is the probability density function of $N(a)$.  We say that a graph $A$ is $d-$\emph{regular} if every neighborhood distribution is uniform on a set of size $d$, \emph{i.e.}, $M_{a,a'}\in\{0,\nicefrac{1}{d}\}$ for all $a,a'\in A$.  A $d-$regular graph is \emph{undirected} if its adjacency matrix is symmetric.  The \emph{edge distribution} on an undirected, $d-$regular graph is the distribution on $A^2$ which draws $a\sim A$, $a'\sim N(a)$ and outputs $(a,a')$.  Likewise, the $k-$\emph{length random walk distribution on} $A$, denoted $\rw_A^k$, is the distribution on $A^k$ which outputs $(a_0,\dots,a_{k-1})$ where $a_1\sim A$ and $a_i\sim N(a_{i-1})$ for $i=2,\dots,k-1$.  For $a\in A$, $\rw_A^k(a)$ is the distribution which outputs a sample from $\rw_A^k$ conditioned on $a_1=a$.  So $\rw_A^k(a)$ outputs a $k-$length random walk in $A$ starting at $a$.

\paragraph{Random Walks on Graphs.} Let $A$ be the vertex set of a graph.  Given $a,a'\in A$, we write $a\sim a'$ if $a$ and $a'$ are connected by an edge.  For $a\in A$, let $N(a)\subset A$ denote the \emph{neighborhood} of $A$, \emph{i.e.}, $N(a):=\{a'\in A:a\sim a'\}$.  For an integer $d\geq1$, we say that $A$ is $d-$regular if $|N(a)|=d$ for all $a\in A$.  For an integer $k\geq1$, let \[\rw_A^k:=\{(a_1,\dots,a_k)\in A^k:a_i\sim a_{i+1}\text{ }\forall\text{ }i=1,\dots,k-1\}\] denote the set of $k-$length random walks in $A$.  Similarly, for $a\in A$, $\rw_A^k(a)$ is the set of $k-$length random walks in $A$ which begin at $a$, so $\rw_A^k(a):=\{(a_1,\dots,a_k)\in\rw_A^k:a_1=a\}$.  We will often view $\rw_A^k$ as a distribution, where $(a_1,\dots,a_k)\sim\rw_A^k$ means that $a_1\sim A$ is drawn uniformly and then $a_{i+1}\sim N(a_i)$ is drawn for $i=1,\dots,k-1$.

\paragraph{Expander Graphs.} Graph expansion is usually defined as the second largest eigenvalue of the graph's adjacency matrix,\footnote{The \emph{adjacency matrix} of the graph $A$ is $M\in\{0,1\}^{|A|\times|A|}$, where $M(a,a')=1$ iff $a\sim a'$.} \emph{i.e.}, \begin{equation}\label{eq:eigen}\lambda:=\max_{x,y\perp\boldone}\frac{|\langle x,My\rangle|}{|x||y|},\end{equation} where the max is over all nonzero $x,y\in\R^{|A|}-\{0\}$ which are perpendicular to the all $1$s vector $\boldone$.  Our Definition~\ref{def:expander} can be recovered from (\ref{eq:eigen}) for any $f,g:A\rightarrow\R$ by setting $x,y\in\R^{|A|}$ to be $x_a=f(a)-\mu_f$ and $y_a=g(a)-\mu_g$.

\paragraph{Cayley Graphs.} Given a finite group $G$ and a subset $U\subseteq G$, the Cayley graph ${\rm Cayley}(G,U)$ has vertex set $G$ with $g\sim g'$ iff $g^{-1}g'\in U$.  Note that ${\rm Cayley}(G,U)$ is $|U|-$regular; additionally, if $U$ is closed under inversion, then ${\rm Cayley}(G,U)$ is undirected.  Cayley graphs play a key role in many explicit constructions of expander graphs.  Ta-Shma's original construction used two Cayley graphs as explicit expander constructions.  The first Cayley graph was over $\F_2^k$, and the second was over ${\sf PGL}_2(\F_q)$, the projective general linear group over a large finite field.  The use of this second Cayley graph put restrictions on some of the parameters, which required some care in order to navigate.  Subsequently to Ta-Shma's original paper, new constructions of expanders based on Cayley graphs have been given.  We will use a new construction, due to Alon~\cite{ALON21}, instead of the ${\sf PGL}_2(\F_q)$ construction as it will give us more flexibility.

\begin{thm}
\label{thm:expander}
We have the following expander constructions from~\cite{ALON21} and~\cite{AGHP92}, respectively.
\begin{enumerate}[align=left]
    \item[{\bf The Outer Graph:}] For all integers $n,d\in\N$ there is an explicit construction of a $d-$regular Cayley graph with $n\cdot(1+o_n(1))$ vertices and expansion $\lambda\leq \frac{8}{\sqrt{d}}$.
    
    \item[{\bf The Inner Graph:}] For all integers $r,\ell\in\N$ such that $\ell\leq r/2$, there exists an explicit\footnote{This Cayley graph construction is actually \emph{fully explicit}, in the sense that given any vertex, the $i-$th neighbor can be computed in polylogarithmic time.} construction of an undirected $2^{2\ell}-$regular Cayley graph over $\F_2^r$ which is a $(r-1)2^{-\ell}-$expander.
\end{enumerate}
\end{thm}

\paragraph{The Shifted Neighborhood Distribution.} Let $B$ be a Cayley graph on $\F_2^{ms}$, and let $d=2^m$.  For any $b=\bigl(b[1],\dots,b[s]\bigr)\in B\cong[d]^s$, let ${\sf shift}(b)=\bigl(b[2],\dots,b[s],b[1]\bigr)\in B$ be the element obtained by circularly shifting the coordinates of $b$.  Given $b\in B$, the \emph{shifted neighborhood distribution of} $b$, denoted $\tilde N(b)$, draws $u\sim U$ (the generator set of the Cayley graph) and outputs ${\sf shift}(b+u)$ (note $b+u$ is a random neighbor of $b$ in $B$).  It is clear that the expansion of $B$ is not affected by using the shifted neighborhood distribution instead of the original neighborhood distribution.  Indeed, \[\Big|\E_{\stackalign{b&\sim B\\ b'&\sim\tilde N(b)}}\bigl[f(b)\cdot g(b')\bigr]-\mu_f\mu_g\Big|=\Big|\E_{\stackalign{b&\sim B\\ b'&\sim N(b)}}\bigl[f(b)\cdot\tilde g(b')\bigr]-\mu_f\mu_{\tilde g}\Big|\leq\lambda\sigma_f\sigma_{\tilde g}=\lambda\sigma_f\sigma_g,\] where $\tilde g=g\circ{\sf shift}$; clearly $(\mu_{\tilde g},\sigma_{\tilde g})=(\mu_g,\sigma_g)$.  Let $\tilde{\rw}^k_B$ denote the set of $k-$length shifted random walks in $B$.  We prove the following claim about $\tilde{\rw}^k_B$, when $k$ is small.
\begin{clm}
\label{clm:bwalk}
For all $k\leq s$, the distribution that chooses $(b_1,\dots,b_k)\sim\tilde\rw^k_B$ and outputs the tuple $( b_1[1],b_2[1],\dots,b_k[1])\in[d]^k$ is identical to the uniform distribution on $[d]^k$.
\end{clm}

\begin{proof}
It suffices to prove the claim for $k=s$, since when $k<s$, the distribution $\tilde\rw_B^k$ is identical to the distribution which draws $(b_1,\dots,b_s)\sim\tilde\rw_B^s$ and outputs $(b_1,\dots,b_k)$.  Note that $\tilde\rw_B^s$ draws $u_1,\dots,u_{s-1}\sim U$, $b_1\sim B$ and outputs $(b_1,\dots,b_s)\in B^s$, where $b_i={\sf shift}(b_{i-1}+u_{i-1})$ for $i=2,\dots,s$.  This means that for all $i=1,\dots,s$, $ b_i[1]=b_1[i]+\sum_{j<i}u_j[i-j+1]$ (addition over $\F_2^m$).  Uniformity of $\big( b_1[1],b_2[1],\dots,  b_t[1]\big)$ follows from the uniformity of $b_1=\bigl(b_1[1],\dots,b_1[s]\bigr)\sim[d]^s$.
\end{proof}

\subsection{The $s$-wide Replacement Product}
Let $A$ and $B$ denote, respectively, the outer and inner graphs promised by Theorem~\ref{thm:expander}.  So $A$ is a $d-$regular graph on (roughly) $n$ vertices, while $B$ is a Cayley graph over $\F_2^{ms}$, where $2^m=d$, so that vertices of $B$ are identified with $s-$tuples of elements in $[d]$: $b=\bigl(b[1],\dots,b[s]\bigr)\in[d]^s$.  Given $a\in A$, a vertex $b\in B$ can be identified with an $s-$tuple of neighbors of $a$ since $|N(a)|=d$.  Define the \emph{rotation map} $\phi:A\times B\rightarrow A$ via $\phi(a,b)=a'$ where $a'$ is the $b[1]-$th neighbor of $a$.  Since $\phi$ only depends on the first coordinate of $b$, we write $\phi(a,\hat b)$ where $\hat b$ is shorthand for $b[1]$.  For any $k\geq1$, the $k-$length $s-$wide replacement walk distribution, denoted $\sw^k_{A,B}$ draws $a\sim A$ and $(b_1,\dots,b_{k-1})\sim\tilde{\rw}^{k-1}_B$, and outputs $(a_1,\dots,a_k)\in A^k$ where $a_1=a$ and $a_{i+1}=\phi(a_i,\hat b_i)$ for $i=1,\dots,k-1$.  Since the graphs $A$ and $B$ will be fixed throughout this paper, we write $\sw^k$ rather than $\sw^k_{A,B}$.  Given $a\in A$, the distribution $\sw^k(a)$ outputs a sample from $\sw^k$ conditioned on $a_1=a$.  Likewise, given $(a,b)\in A\times B$, $\sw^k(a,b)$ outputs a sample from $\sw^k$ conditioned on $(a_1,b_1)=(a,b)$.  The $s-$wide replacement walk is shown in Figure~\ref{fig:1}.

\begin{figure}[htbp]
\begin{center}
 
    \includegraphics[scale=0.4]{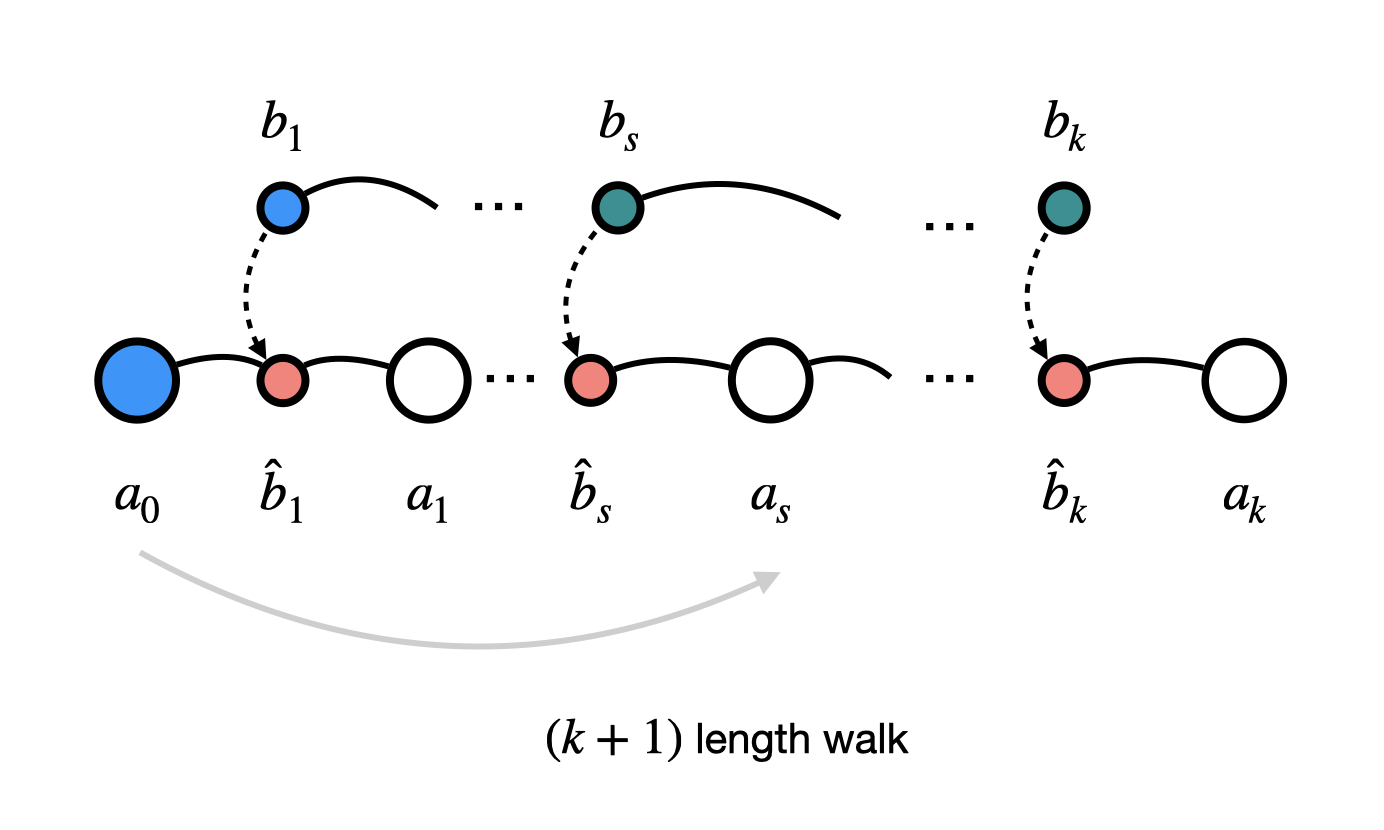}
   %\vspace{-2em}

\caption{Illustration of $s$-wide random walk on $A$ using a graph $B$.}
\label{fig:1}

\end{center}
\end{figure}

\noindent For our graphs $A$ and $B$ (specifically, since $A$ is $d-$regular and $B$ is a Cayley graph over $\F^{ms}_2\cong[d]^s$) the next fact follows immediately from Claim~\ref{clm:bwalk}.

\begin{fact}[{\bf Pseudorandomness}]
\label{fact:pseudorandomness}
For all $k=1,2,\dots,s,s+1$ and all $a\in A$, $\sw^k(a)\equiv\rw^k_A(a)$.
\end{fact}
\noindent
Following Ta-Shma's nomenclature, we will refer to the fact above as the \textit{pseudorandomness} property.
This property will play a crucial role in our proofs below as it will allow us to transform a short $s-$wide walk into a pure random walk on $A$, thus eliminating the dependency on the graph $B$.

\paragraph{Local Invertibility.} Since $A$ is undirected, its edge relation is symmetric.  This means that whenever $a,a'\in A$ and $b\in B$ are such that $a'=\phi(a,\hat b)$, there must exist some $\hat b'\in[d]$ such that $a=\phi(a',\hat b')$.  In this case we say that $(\hat b,\hat b')$ are inverses with respect to the $A-$edge $(a,a')$.  Local invertibility in our context means that these inverse relations are independent of the $A$ edges.  So, specifically, for all $\hat b$ there exists $\hat b'$ such that $(\hat b,\hat b')$ are inverses with respect to all $A$ edges.  This means, for example that for all $a\in A$, if you walk to $a'=\phi(a,\hat b)$ and then continue to $a''=\phi(a',\hat b')$, then $a''=a$.  This property is easy to establish in our situation because $A$ is a Cayley graph.

Practically speaking, what this means for us is that $s-$wide replacement walks can be ``started in the middle''.  For standard random walks, the distribution $\rw_A^k$ which outputs $(a_1,\dots,a_k)$ is identical to the distribution which first chooses $a_i\sim A$ randomly, and then draws $(a_i,a_{i+1},\dots,a_k)\sim\rw_A^{k-i+1}(a_i)$ and $(a_i,a_{i-1},\dots,a_1)\sim\rw_A^i(a_i)$, outputting $(a_1,\dots,a_k)$.  This follows from the regularity of $A$.  Likewise, because of local invertibility, the $s-$wide replacement walk distribution $\sw^k$ is identical to the following ``start in the middle'' version which draws $a_i\sim A$ and $b_i\sim B$, then draws $(b_i,\dots,b_{k-1})\sim\tilde{\rw}_B^{k-i}(b_i)$ and $(b_i,\dots,b_1)\sim\tilde{\rw}_B^i(b_i)$ (in this case the shifted neighborhood distribution needs to shift the other way), then sets $a_{j+1}=\phi(a_j,\hat b_j)$ for $j=i,\dots,k-1$ and $a_{j-1}=\phi(a_j,\hat b_j')$ for $j=i,\dots,2$, where $\hat b_j'$ is the inverse of $\hat b_j$; finally $(a_1,\dots,a_k)$ is output.

\section{Main theorem}
\label{sec:mainthm}
\begin{thm}\label{thm:main}
For every $\ep>0$ there exists an explicit linear code $\{\mathcal{C}_k\}_k$ that has distance $\geq \frac{1}{2}-\ep$ and rate $=\Omega( \ep^{2+o(1)}).$
\end{thm}

\begin{proof}
Fix $k\in\mathbb{N}$.  The construction of $\mathcal{C}_k$ uses the following building blocks.
\begin{itemize}[align=left]
    \item[$\bullet$ {\bf The Base Code:}] Let $\mathcal{C}_0:\{0,1\}^k\rightarrow\{0,1\}^{n_0}$ be an explicit code of bias $\ep_0$ and rate $R_0$.  We use the construction in~\cite{ABNNR92}, so that $R_0=\mathcal{O}(\ep_0^{-3})$.

    \item[$\bullet$ {\bf The Outer Graph:}] Let $A$ be the $d_A-$regular Cayley graph with expansion $\lambda_A$.  We use the construction of Theorem~\ref{thm:expander}, so that $\lambda_A\leq8/\sqrt{d_A}$ and $|A|=n_0\cdot\bigl(1+o_{n_0}(1)\bigr)$.
    
    \item[$\bullet$ {\bf The Inner Graph:}] Let $B$ be a $d_B-$regular Cayley graph over $\F^r_2$  with  expansion $\lambda_B$.  We use the construction of Theorem~\ref{thm:expander} so that $\lambda_B=(r-1)\cdot2^{-\ell}$ and $d_B=2^{2\ell}$ for integers $\ell,r\in\mathbb{N}$ such that $\ell\leq r/2$.
\end{itemize}

\noindent The building blocks carry several parameters which we now connect.  In order to set up the $s-$wide replacement product, define additional parameters $s,m\in\mathbb{N}$ such that $r=ms$, and let $d_A=2^m$, so $B\simeq [d_A]^s$.  It will be important for our analysis to have $\lambda_A\leq\lambda_B^2$; in order to arrange this, set $m=s$ and $\ell=s/5$.  This gives \[\lambda_A\leq\frac{8}{\sqrt{d_A}}=8\cdot2^{-m/2}=\frac{8}{2^{\ell/2}}\cdot2^{-2\ell}\leq(ms-1)^2\cdot2^{-2\ell}=\lambda_B^2,\] where the final inequality holds whenever $s\geq2$.  We will also require $\ep_0\leq\lambda_B/2$ which we ensure by setting $\ep_0=\frac{s^2-1}{2}\cdot2^{-s/5}$.  At this point, all parameters so far have been defined in terms of $s$; we will specify $s$ later.  Note that our setup allows us to use $B$ to take $s-$wide replacement walks in $A$.  We now describe the code.  Given $x\in\{0,1\}^k$, $\mathcal{C}_k(x)$ is computed as follows.
\begin{itemize}
    \item Compute $\mathcal{C}_0(x)\in\{0,1\}^{n_0}$, and define $f:A\rightarrow\{0,1\}$ by setting \[f(a)=\left\{\begin{array}{cc}\mathcal{C}_0(x)_i, & a=\iota(i)\\ 0, & \text{ otherwise}\end{array}\right.\] where $\iota:[n_0]\hookrightarrow A$ is some fixed embedding.
    
    \item Define $g:\sw^t\rightarrow\{0,1\}$ by setting $g(a_0,\dots,a_t)=f(a_0)\oplus\cdots\oplus f(a_t)$.  Output $g\in\{0,1\}^{\sw^t}$.
\end{itemize}

The rate of $\mathcal{C}_k$ is \[{\sf Rate}_k=\frac{k}{|\sw^t|}\geq\frac{k}{|A|}\cdot\frac{1}{|B|}\cdot \frac{1}{d_B^{t-1}}=\Omega(\ep_0^{-3})\cdot2^{-s^2}\cdot d_B^{-(t-1)}=\Omega\bigl(s^{-6}\cdot2^{-s^2}\bigr)\cdot d_B^{-(t-1)}.\]  To bound the bias of $\mathcal{C}_k$, we use the following lemma which is proved in the next section.

\begin{lemma}[{\bf Bias Reduction of Wide Replacement Product Walks}]
\label{lemma:swide}
Let integers $s,t\in\mathbb{N}$ and graphs $A$ and $B$ be as above; so in particular $A$ and $B$ are $\lambda_A$ and $\lambda_B$ expanders with $\lambda_A\leq\lambda_B^2$.  Let $f:A\rightarrow\{0,1\}$ be any function such that $\big|\E_a\bigl[(-1)^{f(a)}\bigr]\big|\leq\lambda_B$.  Then \[\Big|\E_{(a_0,\dots,a_t)\sim\sw^t}\Bigl[(-1)^{f(a_0)\oplus\cdots\oplus f(a_t)}\Bigr]\Big| \leq(2\lambda_B)^{t(1-4/s)}.\]
\end{lemma}

\noindent Note that the function $f:A\rightarrow\{0,1\}$ defined in the first step of computing $\mathcal{C}_k(x)$ satisfies \[\Big|\E_a\bigl[(-1)^{f(a)}\bigr]\Big|\leq2\cdot\Big|\E_{i\sim[n_0]}\bigl[(-1)^{\mathcal{C}_0(x)_i}\bigr]\Big|\leq2\ep_0\leq\lambda_B,\] and so Lemma~\ref{lemma:swide} ensures that ${\sf Bias}_k\leq(2\lambda_B)^{t(1-4/s)}$.  Putting the calculations of ${\sf Rank}_k$ and ${\sf Bias}_k$ together and using $\lambda_B=(s^2-1)/\sqrt{d_B}$ gives \[{\sf Rate}_k=\Omega\Bigl(s^{-6}\cdot(s^2-1)^{-2t}\cdot2^{-2t-s^2+2s/5}\cdot(2\lambda_B)^{8t/s}\Bigr)\cdot{\sf Bias}_k^2=\Omega\Bigl(s^{-5t}\cdot(2\lambda_B)^{8t/s}\Bigr)\cdot{\sf Bias}^2_k,\] where the right most equality holds whenever $6\log s\leq2s/5$ (implied by $s\geq100$) and $t\geq s^2$.  Note, therefore, that for $\eta\in\bigl(0,1/2\bigr)$, ${\sf Rate}_k=\Omega\bigl({\sf Bias}_k^{2+\eta}\bigr)$ holds whenever $(2\lambda_B)^{t(\eta-4\eta/s-8/s)}\leq s^{-5t}$ which, if $\eta\geq24/s$ is implied by $(2\lambda_B)^{\eta/2}\leq s^{-5}$.  Finally, by plugging in $\lambda_B=(s^2-1)\cdot2^{-s/5}$, we see that this holds whenever $\eta s\geq60\log s$.

So finally, let us prove the theorem.  Suppose that we are given $\ep>0$ and $\eta\in\bigl(0,1/2\bigr)$, and we want to construct $\mathcal{C}_k$ such that ${\sf Bias}_k\leq\ep$ and ${\sf Rate}_k=\Omega\bigl({\sf Bias}^{2+\eta}\bigr)$.  We let $\mathcal{C}_k$ be the construction defined above with $s$ chosen large enough so that $\eta s\geq60\log s$; this ensures ${\sf Rate}_k=\Omega\bigl({\sf Bias}_k^{2+\eta}\bigr)$ as noticed above.  Finally, let us choose $t$ large enough so that $t\geq s^2$ and $(2\lambda_B)^{t(1-4/s)}\leq\ep$; this ensures ${\sf Bias}_k\leq\ep$, as desired.
\end{proof}

\section{Proof of Lemma~\ref{lemma:swide}}
In this section we prove the key bias reduction lemma that was the core of Theorem~\ref{thm:main}.  Our proof will be by induction, just like Claim~\ref{clm:genrw}, so we will need to modify the statement of Lemma~\ref{lemma:swide} so it adheres to an inductive argument.

\subsection{Lemma Statement}
Let $A$ and $B$ be the graphs from Section~\ref{sec:mainthm}.  Write $\lambda$ instead of $\lambda_B$ for the expansion of $B$ and recall that $\lambda_A\leq\lambda^2$.  Let $f:A\rightarrow\{0,1\}$ be a function such that $\big|\E_a\bigl[(-1)^{f(a)}\bigr]\big|\leq\lambda$.  For any $k\geq0$, define $g_k:A\times B\rightarrow\mathbb{R}$ by 
\begin{equation}\label{eq:gk}
g_k(a,b)=\E_{(a_0,\dots,a_k)\sim \sw^k(a,b)}\Bigl[(-1)^{f(a_0)\oplus\cdots\oplus f(a_k)}\Bigr].
\end{equation}
Let $\ep_k=\big|\E_{a,b}\bigl[g_k(a,b)\bigr]\big|$ and let $\sigma_k$ be such that $\sigma_k^2+\ep_k^2=\E_{a,b}\bigl[g_k(a,b)^2\bigr]$.  We prove the following.
\begin{lemma}[{\bf Implies Lemma~\ref{lemma:swide}}]
\label{lemma:swideinduciton}
Assume the above setup.   For all $k\geq0$ \[\ep_k\leq(2\lambda)^{k(1-4/s)};\text{  }\sigma_k\leq (2\lambda)^{(k-2)(1-4/s)}.\]
\end{lemma}

\noindent As mentioned, we prove Lemma~\ref{lemma:swideinduciton} by induction.  The following two claims combine to easily prove Lemma~\ref{lemma:swideinduciton}; we will prove them in Sections~\ref{sec:epk} and~\ref{sec:sigmak}.

\begin{clm}[{\bf Base Case.}]
\label{clm:swidebasecase}
Assume the above setup.  For all $k=0,1,\dots,s$:
\[\ep_k\leq\frac{1}{2}\cdot(2\lambda)^{k+1
};\text{ }\sigma_k\leq2\cdot(2\lambda)^{k-1}.\]
\end{clm}

\begin{clm}[{\bf Induction Step.}]
\label{clm:swideinduction}
Assume the above setup.  For all $k>s$:
\begin{center}\begin{minipage}{.9\linewidth}\begin{itemize}
    \item[$\cdot$ $\ep_k\leq$] $\frac{1}{2}(2\lambda)^s(\ep_{k-s}+3\sigma_{k-s})$;
    
    \item[$\cdot$ $\sigma_k^2\leq$]
    $\frac{1}{2}(2\lambda)^{s-2}(\ep_{k-2}+\lambda\sigma_{k-1})\bigl(\ep_{k-s}+(2+\lambda)\sigma_{k-s}\bigr)+\lambda^s\sigma_{k-s}\sigma_{k-1}+\lambda^2\sigma_{k-1}^2$
\end{itemize}\end{minipage}\end{center}
\end{clm}

\begin{proof}[Proof of Lemma~\ref{lemma:swideinduciton}] Claim~\ref{clm:swidebasecase} clearly establishes the base cases since $\frac{1}{2}\cdot(2\lambda)^{k+1}\leq(2\lambda)^{k(1-4/s)}$ and $2\cdot(2\lambda)^{k-1}\leq(2\lambda)^{(k-2)(1-4/s)}$.  For the first part of the induction step, we have \begin{eqnarray*}\ep_k &\leq& \frac{1}{2}\cdot(2\lambda)^s\cdot(\ep_{k-s}+3\sigma_{k-s})\leq\frac{1}{2}\cdot(2\lambda)^s\cdot\Bigl[(2\lambda)^{(k-s)(1-4/s)}+3\cdot(2\lambda)^{(k-s-2)(1-4/s)}\Bigr]\\ &=& 8\lambda^4\cdot\Bigl[(2\lambda)^{k(1-4/s)}+3\cdot(2\lambda)^{(k-2)(1-4/s)}\Bigr]\leq2\lambda^2(4\lambda^2+3)\cdot(2\lambda)^{k(1-4/s)}\leq(2\lambda)^{k(1-4/s)}.\end{eqnarray*} The bound $2\lambda^2(4\lambda^2+3)\leq1$ holds because $\lambda\leq1/3$.  The second part of the induction step is similar:
\begin{eqnarray*}
\sigma_k^2
    &\leq&
\frac{1}{2}\cdot(2\lambda)^{s-2}\cdot(\ep_{k-2}+\lambda\sigma_{k-1})\bigl(\ep_{k-s}+(2+\lambda)\sigma_{k-s}\bigr)+\lambda^s\sigma_{k-s}\sigma_{k-1}+\lambda^2\sigma_{k-1}^2\\
    &\leq&
\frac{1}{2}\cdot(2\lambda)^2\cdot\Bigl[(2\lambda)^{(k-2)(1-4/s)}+\lambda(2\lambda)^{(k-3)(1-4/s)}\Bigr]\cdot\Bigl[(2\lambda)^{k(1-4/s)}+(2+\lambda)(2\lambda)^{(k-2)(1-4/s)}\Bigr]+\\
    &+&
\lambda^s(2\lambda)^{(k-s-2)(1-4/s)}(2\lambda)^{(k-3)(1-4/s)}+\lambda^2(2\lambda)^{2(k-3)(1-4/s)}\\
    &=&
2\lambda^2(2\lambda)^{(2k-2)(1-4/s)}+2\lambda^3(2\lambda)^{(2k-3)(1-4/s)}+(4\lambda^2+2\lambda^3)(2\lambda)^{(2k-4)(1-4/s)}+\\
    &+&
(4\lambda^3+2\lambda^4)(2\lambda)^{(2k-5)(1-4/s)}+2^{4-s}\lambda^4(2\lambda)^{(2k-5)(1-4/s)}+\lambda^2(2\lambda)^{(2k-6)(1-4/s)}\\
    &\leq&
\biggl[2\lambda^2+2\lambda^3+(4\lambda^2+2\lambda^3)+(2\lambda^2+\lambda^3)+2^{3-s}\lambda^3+\frac{1}{4}\biggr]\cdot(2\lambda)^{(2k-4)(1-4/s)}\leq(2\lambda)^{(2k-4)(1-4/s)},\end{eqnarray*}
where the last bound has used $8\lambda^2+6\lambda^3\leq3/4$ which holds because $\lambda\leq1/4$.
\end{proof}

\subsection{Key Intuition}
\label{sec:intuition}
In this section we zoom in on some of the key steps in the coming proofs in order to give extra explanations and intuitions.

\paragraph{$s-$wide Replacement Product Walks in $A$.} Recall that a random $s-$wide replacement product walk in $A$ (\emph{i.e.}, a random sample from $\sw^k$) is produced as follows:
\begin{enumerate}
    \item choose base points $(a,b)\sim A\times B$;
    \item generate $(b_1,\dots,b_k)\in B^k$ as follows:
    \begin{itemize}
        \item[$(i)$] set $b_1=b$;
        \item[$(ii)$] for $i\geq2$, draw $b_i\sim N(b_{i-1})$ and set $b_i={\sf shift}(b_i)$, where ${\sf shift}$ cycles the coordinates of an element of $B\simeq[d]^s$, so ${\sf shift}\bigl(b_i[1],\dots,b_i[s]\bigr)=\bigl(b_i[2],\dots,b_i[s],b_i[1]\bigr)$.
    \end{itemize}
    \item generate and output $(a_0,\dots,a_k)\in A^{k+1}$ as follows:
    \begin{itemize}
        \item[$(i)$] set $a_0=a$;
        \item[$(ii)$] for $i\geq1$, set $a_i=\phi(a_{i-1},\hat b_i)$ where $\hat b_i=b_i[1]\in[d]$ denotes the first coordinate of $b_i\in[d]^s$, and where $\phi$ is the rotation map of $A$.
    \end{itemize}
\end{enumerate}

\paragraph{Pseudorandomness.} As mentioned in Section~\ref{sec:prelims}, when $k\leq s$ the distributions $\sw^k$ and $\rw_A^{k+1}$ are identical.  That is, a random $k-$step $s-$wide replacement product walk in $A$ is just a random $(k+1)-$step random walk in $A$.  The following is an example of how this concept manifests itself in the next section.  Let $\ep_k(a)=\E_b\bigl[g_k(a,b)\bigr]$. \[\ep_k(a)=\E_{(a_0,\dots,a_k)\sim\sw^k(a)}\Bigl[(-1)^{f(a_0)\oplus\cdots\oplus f(a_k)}\Bigr]=\E_{(a_0,\dots,a_k)\sim\rw^{k+1}_A}\Bigl[(-1)^{f(a_0)\oplus\cdots\oplus f(a_k)}\Bigr]=h_{k+1}(a),\] whenever $k\leq s$, where $h_{k+1}$ is the function defined and analyzed in Claim~\ref{clm:exprw}.

\paragraph{The Ignore First Step Trick.} This refers to a key step in the proof that for all $k\geq1$, \begin{equation}\label{eq:firststeptrick}\sigma_k^2\leq\E_a\bigl[\ep_{k-1}(a)^2\bigr]+\lambda^2\sigma_{k-1}^2.\end{equation}  This bound is useful as it reduces the task of bounding $\sigma_k^2$ to the task of bounding $\E_a\bigl[\ep_{k-1}(a)^2\bigr]$, which will turn out to be much easier.  The proof of (\ref{eq:firststeptrick}) requires other ideas as well.  Recall from the previous paragraph the definition of $\ep_k(a)$; additionally let $\sigma_k(a)$ be such that $\sigma_k(a)^2+\ep_k(a)^2=\E_b\bigl[g_k(a,b)^2\bigr]$. \begin{eqnarray*} \sigma_k^2 &\leq& \sigma_k^2+\ep_k^2=\E_{a,b}\bigl[g_k(a,b)^2\bigr]=\E_{a,b}\Bigl[\E_{b'\sim N(b)}\bigl[g_{k-1}(a',b')\bigr]^2\Bigr]=\E_{\stackalign{&a\sim A\\b&\sim_{B^2}b'}}\bigl[g_{k-1}(a,b)\cdot g_{k-1}(a,b')\bigr]\\ &\leq&\E_a\bigl[\ep_{k-1}(a)^2\bigr]+\lambda^2\E_a\bigl[\sigma_{k-1}(a)^2\bigr]\leq\E_a\bigl[\ep_{k-1}(a)^2\bigr]+\lambda^2\sigma_{k-1}^2.\end{eqnarray*}  The second equation on the first line holds because $g_k(a,b)=(-1)^{f(a)}\cdot\E_{b'\sim N(b)}\bigl[g_{k-1}(a',b')\bigr]$, where $a'=\phi(a,\hat b)$; the first inequality on the second line follows from the expander mixing lemma (Definition~\ref{def:expander}) on $B^2$ (a $\lambda^2-$expander); the final inequality has used $\E_a\bigl[\sigma_{k-1}(a)^2\bigr]\leq\sigma_{k-1}^2$ which holds because \[\E_a\bigl[\sigma_{k-1}(a)^2+\ep_{k-1}(a)^2\bigr]=\E_{a,b}\bigl[g_{k-1}(a,b)^2\bigr]=\sigma_{k-1}^2+\ep_{k-1}^2,\] and $\ep_{k-1}^2\leq\E_a\bigl[\ep_{k-1}(a)^2\bigr]$ (Jensen's inequality).  The ignore first step trick is the reasoning behind the final equation on the first line.  The observation is that the distribution which draws $(a,b)\sim A\times B$ and $b',b''\sim N(b)$ and outputs $(a',b',b'')$ where $a'=\phi(a,\hat b)$ is identical to the distribution which draws $a'\sim A$ and a random edge $b'\sim_{B^2}b''$ in $B^2$ and outputs $(a',b',b'')$.  See Figure~\ref{fig:2} for intuition.  

\begin{figure}[h]
\begin{center}

    \includegraphics[scale=0.3]{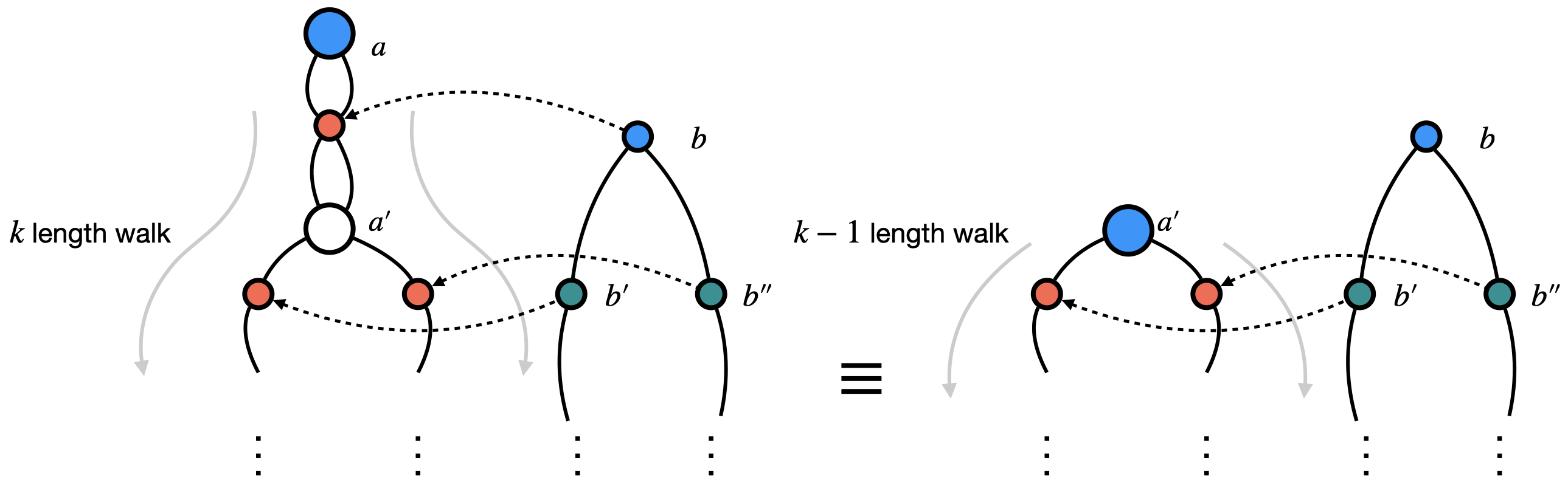}

\caption{``Ignore first step'' trick.}
\vspace{-1.5 em}
 \label{fig:2}

\end{center}
\end{figure}

\paragraph{Starting the Replacement Walk in the Middle.} A useful feature of random walks on an undirected $d-$regular graph is that the steps can be generated out of order.  Specifically, the vertices in a $k-$step random walk can be generated by choosing $a_i\sim A$ first for any $i\in[k]$ and then drawing two walks $(a_i,a_{i+1},\dots,a_k)\sim\rw_A^{k-i+1}(a_i)$, $(a_i,a_{i-1},\dots,a_1)\sim\rw_A^i(a_i)$ and outputting $(a_1,\dots,a_k)$.  Replacement product walks also have this feature, though correctly formulating it requires precision.  We will use that the following distribution is identical to $\sw^k$ for any $i\in\{0,1\dots,k-1\}$:
\begin{enumerate}
    \item $a_i\sim A$ and a random edge $b_i\sim b_{i+1}$ in $B$; set $b_{i+1}={\sf shift}(b_{i+1})$;
    \item generate $(b_1,\dots,b_k)\in B^k$ as follows:
    \begin{itemize}
        \item[$(i)$] for $j\geq i+2$, draw $b_j\sim N(b_{j-1})$ and set $b_j={\sf shift}(b_j)$;
        \item[$(ii)$] for $j\leq i-1$, draw $b_j\sim N(b_{j+1})$ and set $b_j={\sf shift}^{-1}(b_j)$;
    \end{itemize}
    \item generate and output $(a_0,\dots,a_k)\in A^{k+1}$ as follows:
    \begin{itemize}
        \item[$(i)$] for $i\geq i+1$, set $a_i=\phi(a_{i-1},\hat b_i)$ where $\hat b_i=b_i[1]\in[d]$ denotes the first coordinate of $b_i\in[d]^s$, and where $\phi$ is the rotation map of $A$;
        \item[$(ii)$] for $j\leq i-1$, set $a_j=\phi^{-1}(a_{j+1},\hat b_j)$ where $\phi^{-1}(a,\hat b)=\phi(a,\hat b')$ where $\hat b'$ is the local inverse of $\hat b$.
    \end{itemize}
\end{enumerate}
An example of how this is used is the first step of the bound for $\ep_k$ when $k>s$:

\begin{eqnarray*}\ep_k &=& \bigg|\E_{(a_0,\dots,a_k)\sim\sw^k}\Bigl[(-1)^{f(a_s)}\cdot(-1)^{f(a_0)\oplus\cdots\oplus f(a_s)}\cdot(-1)^{f(a_s)\oplus\cdots\oplus f(a_k)}\Bigr]\bigg|\\ &=& \bigg|\E_{\stackalign{a_s&\sim A\\b_s&\sim b_{s+1}}}\Bigl[(-1)^{f(a_s)}\cdot\backwardsvec{g}_s(a_s,b_s)\cdot g_{k-s}(a_s,b_{s+1})\Bigr]\bigg|,\end{eqnarray*} where $\backwardsvec{g}_s(a,b)$ indicates that the repalcement walk is drawn in the ``backwards'' fashion according to Steps 2(ii) and 3(ii) above.  Equivalently, $\backwardsvec{g}_s(a,b)$ is the expectation of $(-1)^{f(a_0)\oplus\cdots\oplus f(a_s)}$ over $(a_0,\dots,a_s)\sim\sw^s$ conditioned on $(a_s,b_s)=(a,b)$.

\begin{figure}[h]
\begin{center}
 
    \includegraphics[scale=0.3]{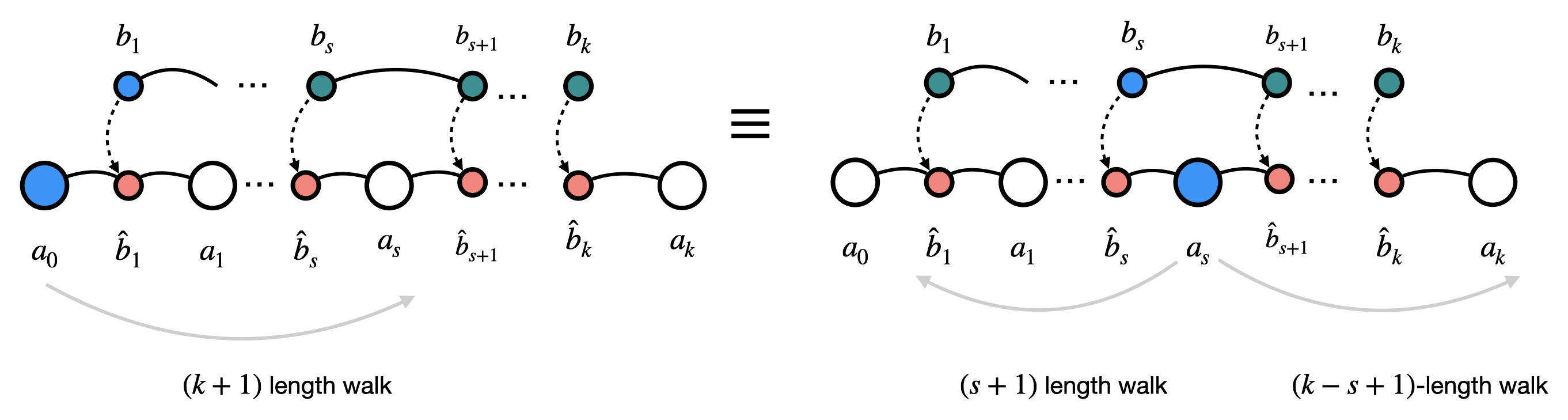}

\caption{Starting the Replacement Walk in the Middle.}
\vspace{-2 em}
 \label{fig:3}

\end{center}
\end{figure}

\subsection{Bounding the $\ep_k$ Terms}
\label{sec:epk}
In this section we bound the $\ep_k$ terms in Claims~\ref{clm:swidebasecase} and~\ref{clm:swideinduction}, thereby proving half of each claim.  We bound the $\sigma_k$ terms in the next section.

\paragraph{The Base Case.} This follows directly from the pseudorandomness property, and the analysis already done in Section~\ref{sec:techniques} (Claim~\ref{clm:exprw}).  Specifically, when $k\leq s$, we have \[\ep_k=\Big|\E_a\bigl[\ep_k(a)\bigr]\Big|=\Big|\E_a\bigl[h_{k+1}(a)\bigr]\Big|\leq\frac{1}{2}\cdot(2\lambda)^{k+1},\] where $\ep_k(a)=h_{k+1}(a)$ by pseudorandomness ($h_{k+1}$ is the function defined in Claim~\ref{clm:exprw}).

\paragraph{The Induction Step.} Fix $k>s$.  We have \[\ep_k =\bigg|\E_{\stackalign{a&\sim A\\b&\sim b'}}\Bigl[(-1)^{f(a)}\cdot\backwardsvec{g}_s(a,b)\cdot g_{k-s}(a,b')\Bigr]\bigg|\leq \bigg|\E_{a\sim A}\Bigl[(-1)^{f(a)}\cdot\backwardsvec{\ep}_s(a)\cdot\ep_{k-s}(a)\Bigr]\bigg|+\lambda\sigma_s\sigma_{k-s},\] where the equality holds by starting the replacement walk in the middle, and the inequality is the expander mixing lemma (Definition~\ref{def:expander}) on $B$.  We are using the shorthand $\backwardsvec{\ep}_s(a)$ for $\E_b\bigl[\backwardsvec{g}_s(a,b)\bigr]$, and we have used Cauchy-Schwarz to bound the standard deviation terms, just as we did in the computation in the ``ignore first step trick'' paragraph in Section~\ref{sec:intuition}.  Specifically, \[\E_a\bigl[\backwardsvec{\sigma}_s(a)\cdot\sigma_{k-s}(a)\bigr]\leq\sqrt{\E_a[\backwardsvec{\sigma}_s(a)^2]}\sqrt{\E_a[\sigma_{k-s}(a)^2]}\leq\backwardsvec{\sigma}_s\sigma_{k-s}=\sigma_s\sigma_{k-s}.\]  By pseudorandomness, $(-1)^{f(a)}\cdot\backwardsvec{\ep}_s(a)=(-1)^{f(a)}\cdot h_{s+1}(a)=\E_{a'\sim N(a)}\bigl[h_s(a')\bigr]=\E_{a'\sim N(a)}\bigl[\ep_{s-1}(a')\bigr]$, and so we get the desired bound on $\ep_k$ via the expander mixing lemma on $A$, as follows: \begin{eqnarray*}\ep_k &\leq& \Big|\E_{a\sim a'}\bigl[\ep_{s-1}(a)\cdot\ep_{k-s}(a')\bigr]\Big|+\lambda\sigma_s\sigma_{k-s}\leq\ep_{s-1}\ep_{k-s}+\lambda^2\sigma_{s-1}\sigma_{k-s}+\lambda\sigma_s\sigma_{k-s}\\ &\leq& \frac{1}{2}(2\lambda)^s(\ep_{k-s}+3\sigma_{k-s}).\end{eqnarray*}

\subsection{Bounding the $\sigma_k$ Terms}
\label{sec:sigmak}

\paragraph{The Base Case.} We have already noted that when $1\leq k\leq s$, $\ep_{k-1}(a)=h_k(a)$ by pseudorandomness.  Thus, $\E_a\bigl[\ep_{k-1}(a)^2\bigr]=\E_a\bigl[h_k(a)^2\bigr]\leq(2\lambda)^{2k-2}$, by Claim~\ref{clm:exprw}.  It follows from the first step trick that $\sigma_k^2\leq(2\lambda)^{2k-2}+\lambda^2\sigma_{k-1}^2$, which implies $\sigma_k\leq(2\lambda)^{k-1}+\lambda\sigma_{k-1}$.  Iterating this bound gives \[\sigma_k\leq\lambda^{k-1}\cdot\bigl(2^{k-1}+2^{k-2}+\cdots+2+1\bigr)\leq2\cdot(2\lambda)^{k-1}.\]

\paragraph{The Induction Step.} Fix $k>s$.  As mentioned in the ``ignore first step trick'' paragraph in Section~\ref{sec:intuition}, $\sigma_k^2\leq\E_a\bigl[\ep_{k-1}(a)^2\bigr]+\lambda^2\sigma_{k-1}^2$ holds and so it suffices to bound $\E_a\bigl[\ep_{k-1}(a)^2\bigr]$.  By starting the replacement walk in the middle, we get \[\E_a\bigl[\ep_{k-1}(a)^2\bigr]=\E_{\stackalign{a_{s-1}&\sim A\\ b_{s-1}&\sim b_s}}\Bigl[(-1)^{f(a_{s-1})}\cdot g_{k-s}(a_{s-1},b_s)\cdot G(a_{s-1},b_{s-1})\Bigr],\] where $G:A\times B\rightarrow\R$ is defined by $G(a,b):=\E_{(a_0,\dots,a_{s-1})}\bigl[(-1)^{f(a_{s-1})\oplus\cdots\oplus f(a_0)}\cdot\ep_{k-1}(a_0)\bigr]$, where the expectation is over $(a_0,\dots,a_{s-1})$ drawn as follows:
\begin{itemize}
    \item[$\cdot$] set $b_{s-1}=b$; for $1\leq i\leq s-2$, draw $b_i\sim N(b_{i+1})$ and then set $b_i={\sf shift}^{-1}(b_i)$;
    \item[$\cdot$] set $a_{s-1}=a$; for $0\leq i\leq s-2$ set $a_i=\phi^{-1}(a_{i+1},\hat b_{i+1})$.
\end{itemize}
The expander mixing lemma (Definition~\ref{def:expander}) on $B$ gives \[\E_a\bigl[\ep_{k-1}(a)^2\bigr]\leq\E_a\Bigl[(-1)^{f(a)}\cdot\ep_{k-s}(a)\cdot\mu_G(a)\Bigr]+\lambda\sigma_{k-s}\sigma_G,\] where $\mu_G:=\E_{a,b}\bigl[G(a,b)\bigr]$, $\mu_G(a):=\E_b\bigl[G(a,b)\bigr]$ and $\sigma_G$ is such that $\sigma_G^2+\mu_G^2=\E_{a,b}\bigl[G(a,b)^2\bigr]$.  By pseudorandomness, $\mu_G(a)=\E_{(a_0,\dots,a_{s-1})\sim\rw^s_A(a)}\bigl[(-1)^{f(a_0)\oplus\cdots\oplus f(a_{s-1})}\cdot\ep_{k-1}(a_{s-1})\bigr]=\hat h_s(a)$, where $\hat h_s:A\rightarrow\R$ is given by $\hat h_s(a)=\E_{(a_1,\dots,a_s)\sim\rw_A^s}\bigl[(-1)^{f(a_1)\oplus\cdots\oplus f(a_s)}\cdot\ep_{k-1}(a_s)\bigr]$.  Note this is the function defined in Claim~\ref{clm:genrw}, instantiated with $H(a)=\ep_{k-1}(a)$.  We have $(-1)^{f(a)}\cdot\mu_G(a)=\E_{a'\sim N(a)}\bigl[\hat h_{s-1}(a')\bigr]$, and so by the expander mixing lemma on $A$ and Claim~\ref{clm:genrw} we have \begin{eqnarray*}\E_a\bigl[\ep_{k-1}(a)^2\bigr]&\leq&\E_{a\sim a'}\bigl[\ep_{k-s}(a)\cdot\hat h_{s-1}(a')\bigr]+\lambda\sigma_{k-s}\sigma_G\\ &\leq& \ep_{k-s}\cdot2^{s-3}(\lambda^{s-2}\cdot\hat\ep_1+\lambda^{s-1}\hat\sigma_1)+\lambda^2\sigma_{k-s}\cdot2^{s-3}(\lambda^{s-3}\hat\ep_1+\lambda^{s-2}\hat\sigma_1)+\lambda\sigma_{k-s}\sigma_G,\end{eqnarray*} where $\hat\ep_1$ and $\hat\sigma_1$ are the notations from Claim~\ref{clm:genrw}.  In our case, $\hat\ep_1=\E_a\bigl[(-1)^{f(a)}\cdot\ep_{k-1}(a)\bigr]=\ep_{k-2}$, and $\hat\sigma_1=\sqrt{\E_a[\ep_{k-1}(a)^2]-\hat\ep_1^2}\leq\sqrt{\E_{a,b}[g_{k-1}(a,b)^2]-\hat\ep_1^2}=\sqrt{\sigma_{k-1}^2+\ep_{k-1}^2-\ep_{k-2}^2}\leq\sigma_{k-1}$.  We have used Jensen's inequality and that $\ep_{k-2}\geq\ep_{k-1}$.  Using these values and remembering the bound $\sigma_k^2\leq\E_a\bigl[\ep_{k-1}(a)^2\bigr]+\lambda^2\sigma_{k-1}^2$ gives \begin{equation}\label{eq:almostdone}\sigma_k^2\leq\frac{1}{2}(2\lambda)^{s-2}(\ep_{k-2}+\lambda\sigma_{k-1})(\ep_{k-s}+\lambda\sigma_{k-s})+\lambda\sigma_{k-s}\sigma_G+\lambda^2\sigma_{k-1}^2.\end{equation}  This is almost the required bound except we still need to simplify $\sigma_G$.  For this purpose, let us add a parameter to our notation for $G$, writing $G_{s-1}$ instead of $G$, since it is an expectation over a length $(s-1)$ ``backwards'' replacement walk.  For $r\leq s-1$, let $\mu_r:=\E_{a,b}\bigl[G_r(a,b)\bigr]$, let $\mu_r(a):=\E_b\bigl[G_r(a,b)\bigr]$ and $\tau_r$ such that $\tau_r^2+\mu_r^2=\E_{a,b}\bigl[G_r(a,b)^2\bigr]$.  We need to bound$\tau_{s-1}$.  By the ignore first step trick and expander mixing lemma on $B^2$, \[\tau_{s-1}^2\leq\E_{a,b}\bigl[G_{s-1}(a,b)^2\bigr]=\E_{\stackalign{a&\sim A\\ b&\sim_{B^2}                 b'}}\Bigl[G_{s-2}(a,b)\cdot G_{s-2}(a,b')\Bigr]\leq\E_a\bigl[\mu_{s-2}(a)^2\bigr]+\lambda^2\tau_{s-2}^2.\]  We have already seen that $\mu_{s-2}(a)=\hat h_{s-1}(a)$, and so by Claim~\ref{clm:genrw} and our computation of $\hat\ep_1$ and $\hat\sigma_1$ above,  $\tau_{s-1}^2\leq(2\lambda)^{2s-6}(\ep_{k-2}+\lambda\sigma_{k-1})^2+\lambda^2\tau_{s-2}^2$, which implies $\tau_{s-1}\leq(2\lambda)^{s-3}(\ep_{k-2}+\lambda\sigma_{k-1})+\lambda\tau_{s-2}$.  Iterating this bound (and using $\tau_0\leq\sigma_{k-1}$) gives \[\tau_{s-1}\leq\lambda^{s-3}(\ep_{k-2}+\lambda\sigma_{k-1})(2^{s-3}+2^{s-4}+\cdots)+\lambda^{s-1}\tau_0\leq2\cdot(2\lambda)^{s-3}(\ep_{k-2}+\lambda\sigma_{k-1})+\lambda^{s-1}\sigma_{k-1}.\]  Plugging this into (\ref{eq:almostdone}) gives the desired bound: \[\sigma_k^2\leq\frac{1}{2}(2\lambda)^{s-2}(\ep_{k-2}+\lambda\sigma_{k-1})\bigl(\ep_{k-s}+(2+\lambda)\sigma_{k-s}\bigr)+\lambda^s\sigma_{k-s}\sigma_{k-1}+\lambda^2\sigma_{k-1}^2.\]

\

\section{Expander Hitting Set Lemma}
Just for fun, we include a new proof of the classical expander hitting set lemma.

\begin{lemma}
\label{lem:hitting}
Let $A$ be a $\lambda-$expander, and let $S\subset A$ be a set of size $|S|=\rho|A|$.  Then for all $t\geq1$, \[\prob{(a_1,\dots,a_t)\sim\rw^t}\Bigl[a_i\in S\text{ }\forall\text{ }i=1,\dots,t\Bigr]\leq\rho\cdot\bigl(\rho+\lambda(1-\rho)\bigr)^{t-1}.\]
\end{lemma}

\begin{proof}
Let $\mathbbm{1}_S:A\rightarrow\{0,1\}$ be the indicator function of $S$.  For $k\geq1$, define $g_k:A\rightarrow\mathbb{R}$ by \[g_k(a)=\prob{(a_1,\dots,a_k)\sim\rw^k(a)}\Bigl[a_i\in S\text{ }\forall\text{ }i=1,\dots,k\Bigr].\] Let $\ep_k:=\E_a\bigl[g_k(a)\bigr]$ and $\sigma_k$ be so $\sigma_k^2+\ep_k^2=\E_a\bigl[g_k(a)^2\bigr]$.  Our proof is by induction on $t$; it is clear that the lemma holds in the base case.  For $k\geq2$, note that $g_k(a)=\mathbbm{1}_S(a)\cdot\E_{a'\sim N(a)}\bigl[g_{k-1}(a')\bigr]$ holds, and so  \[\sigma_k^2+\ep_k^2=\E_a\bigl[g_k(a)^2\bigr]=\E_{\stackalign{a&\sim A\\a',&a''\sim N(a)}}\Bigl[\mathbbm{1}_S(a)\cdot g_{k-1}(a')\cdot g_{k-1}(a'')\Bigr]=\ep_{2k-1}.\]  We have used that $\mathbbm{1}_S(a)^2=\mathbbm{1}_S(a)$ holds for all $a\in A$, and that choosing $a\sim A$ and then two $(k-1)$ length walks starting at $a$ is identical to simply choosing a random walk of length $(2k-1)$.  Now, fix $t\geq2$ and $k,\ell\geq1$ such that $t=k+\ell$.  We have \begin{eqnarray*}
    \ep_t
        &=&
    \E_{(a_1,\dots,a_t)\sim\rw^t}\Bigl[\mathbbm{1}_S(a_1)\cdots\mathbbm{1}_S(a_t)\Bigr]=\E_{a\sim a'}\bigl[g_k(a)\cdot g_\ell(a')\bigr]\leq\ep_k\ep_\ell+\lambda\sigma_k\sigma_\ell\\
        &\leq&
    \sqrt{\ep_k^2+\lambda\sigma_k^2}\cdot\sqrt{\ep_\ell^2+\lambda\sigma_\ell^2}=\sqrt{(1-\lambda)\ep_k^2+\lambda\ep_{2k-1}}\cdot\sqrt{(1-\lambda)\ep_\ell^2+\lambda\ep_{2\ell-1}}
\end{eqnarray*}
where the last inequality on the first line is the expander mixing lemma on $A$ and the first inequality on the second line is Cauchy-Schwarz.  Note that if $2k-1<t$ then we can use induction to bound the terms on the right hand side: \[(1-\lambda)\ep_k^2+\lambda\ep_{2k-1}\leq\rho\cdot\bigl(\rho+\lambda(1-\rho)\bigr)^{2k-2}\cdot\bigl[(1-\lambda)\rho+\lambda\bigr]=\rho\cdot\bigl(\rho+\lambda(1-\rho)\bigr)^{2k-1}.\]  Therefore, if $t$ is even, we can set $k=\ell=t/2$ to obtain $\ep_t\leq\rho\cdot\bigl(\rho+\lambda(1-\rho)\bigr)^{t-1}$, as desired.  This does not fully work if $t$ is odd since if we set $k=\big\lceil t/2\big\rceil$ and $\ell=\big\lfloor t/2\big\rfloor$, then $2k-1=t$ and so we cannot use induction to bound $\ep_{2k-1}$.  However, we can bound $\ep_k$, $\ep_\ell$, $\ep_{2\ell-1}$ by induction; this gives \[\ep_t^2\leq\Bigl((1-\lambda)\rho^2\bigl(\rho+\lambda(1-\rho)\bigr)^{2k-2}+\lambda\ep_t\Bigr)\cdot\Bigl(\rho\bigl(\rho+\lambda(1-\rho)\bigr)^{2\ell-1}\Bigr)=2A\cdot\ep_t+B,\] where $A=\frac{\lambda\rho}{2}\cdot\bigl(\rho+\lambda(1-\rho)\bigr)^{t-2}$ and $B=(1-\lambda)\rho^3\bigl(\rho+\lambda(1-\rho)\bigr)^{2t-3}$.  Collecting the terms in this way allows us to proceed by completing the square.  We get $\ep_t\leq A+\sqrt{A^2+B}$ and we complete the proof by showing that $A+\sqrt{A^2+B}=\rho\bigl(\rho+\lambda(1-\rho)\bigr)^{t-1}$.  For this last calculation, set the shorthand $\Phi:=\rho+\lambda(1-\rho)$.  We have \[A+\sqrt{A^2+B}=\rho\cdot\Phi^{t-2}\cdot\biggl[\frac{\lambda}{2}+\sqrt{\frac{\lambda^2}{4}+\rho(1-\lambda)\Phi}\biggr]=\rho\cdot\Phi^{t-1},\] where the final equation holds because $\Phi=\lambda/2+\sqrt{\lambda^2/4+\rho(1-\lambda)\Phi}$, which is verified by a simple calculation.
\end{proof}

\section*{Acknowledgement}
The authors would like to thank Prahladh Harsha and Aparna Shankar for many helpful discussions.

\bibliographystyle{alpha}
\bibliography{refs}

\newcommand{\etalchar}[1]{$^{#1}$}
\begin{thebibliography}{ABN{\etalchar{+}}92}

\bibitem[ABN{\etalchar{+}}92]{ABNNR92}
Noga Alon, Jehoshua Bruck, Joseph Naor, Moni Naor, and Ron~M Roth.
\newblock Construction of asymptotically good low-rate error-correcting codes
  through pseudo-random graphs.
\newblock {\em IEEE Transactions on information theory}, 38(2):509--516, 1992.

\bibitem[AGHP92]{AGHP92}
Noga Alon, Oded Goldreich, Johan H{\aa}stad, and Ren{\'{e}} Peralta.
\newblock Simple construction of almost k-wise independent random variables.
\newblock {\em Random Struct. Algorithms}, 3(3):289--304, 1992.

\bibitem[Alo21]{ALON21}
Noga Alon.
\newblock Explicit expanders of every degree and size.
\newblock {\em Combinatorica}, pages 1--17, 2021.

\bibitem[BATS11]{BT11}
Avraham Ben-Aroya and Amnon Ta-Shma.
\newblock A combinatorial construction of almost-ramanujan graphs using the
  zig-zag product.
\newblock {\em SIAM Journal on Computing}, 40(2):267--290, 2011.

\bibitem[CJW19]{ChenW19}
Lijie Chen, Ce~Jin, and R~Ryan Williams.
\newblock Hardness magnification for all sparse np languages.
\newblock In {\em 2019 IEEE 60th Annual Symposium on Foundations of Computer
  Science (FOCS)}, pages 1240--1255. IEEE, 2019.

\bibitem[DK17]{DK17}
Irit Dinur and Tali Kaufman.
\newblock High dimensional expanders imply agreement expanders.
\newblock In {\em 2017 IEEE 58th Annual Symposium on Foundations of Computer
  Science (FOCS)}, pages 974--985. IEEE, 2017.

\bibitem[Gil52]{Gil52}
E.~N. Gilbert.
\newblock A comparison of signalling alphabets.
\newblock {\em The Bell System Technical Journal}, 31(3):504--522, 1952.

\bibitem[Plo60]{PLOT60}
Morris Plotkin.
\newblock Binary codes with specified minimum distance.
\newblock {\em IRE Transactions on Information Theory}, 6(4):445--450, 1960.

\bibitem[Sha79]{SHAMIR79}
Adi Shamir.
\newblock How to share a secret.
\newblock {\em Communications of the ACM}, 22(11):612--613, 1979.

\bibitem[STV01]{STV01}
Madhu Sudan, Luca Trevisan, and Salil Vadhan.
\newblock Pseudorandom generators without the xor lemma.
\newblock {\em Journal of Computer and System Sciences}, 62(2):236--266, 2001.

\bibitem[TS17]{TaShma17}
Amnon Ta-Shma.
\newblock Explicit, almost optimal, epsilon-balanced codes.
\newblock In {\em Proceedings of the 49th Annual ACM SIGACT Symposium on Theory
  of Computing}, pages 238--251, 2017.

\bibitem[Var57]{Var57}
R.~R. Varshamov.
\newblock Estimate of the number of signals in error correcting codes.
\newblock {\em Docklady Akad. Nauk, S.S.S.R.}, 117:739--741, 1957.

\end{thebibliography}

\appendix

\end{document}